\documentclass[11pt,letter,onecolumn]{article}
\usepackage[top=2cm, bottom=2cm, left=2cm, right=2cm]{geometry}
\usepackage{times}
\usepackage{amsmath}
\usepackage{amssymb}
\usepackage{amsthm}
\usepackage{algorithm}
\usepackage{algorithmic}
\usepackage{color, array, colortbl}
\usepackage{graphicx}
\usepackage{multirow}
\usepackage{wrapfig}

\usepackage{framed}
\usepackage{color}
\definecolor{shadecolor}{rgb}{0.9,0.9,0.9}

\newcommand{\COMMENTED}[1]{}

\newtheorem{theorem}{Theorem}[section]

\newtheorem{lemma}[theorem]{Lemma}
\newtheorem{remark}[theorem]{Remark}

\newtheorem{definition}[theorem]{Definition}

\usepackage{algorithm}
\usepackage{algorithmic}

\newcommand{\RULEZERO}{\textsc{Axiom}~0}
\newcommand{\RULEONE}{\textsc{Axiom}~1}
\newcommand{\RULETWO}{\textsc{Axiom}~2}
\newcommand{\RULETHREE}{\textsc{Axiom}~3}
\newcommand{\diam}{\textsc{Diam}}
\newcommand{\degree}{\textsc{Deg}}
\newcommand{\dist}{d}
\newcommand{\Comp}{\textsc{Comp}}
\newcommand{\map}{\textsc{Map}}
\newcommand{\cancel}[1]{}

\newcommand{\ALG}{\textsc{Alg}}

\begin{document}

\title{Misleading Stars: What Cannot Be Measured in the Internet?}

\author {
   Yvonne-Anne Pignolet$^1$, Stefan Schmid$^2$, Gilles Tredan$^2$\\
   \small $^1$ ABB Research, Switzerland; yvonne-anne.pignolet@ch.abb.com\\
   \small $^2$ Deutsche Telekom Laboratories \& TU Berlin, Germany; \{stefan,gilles\}@net.t-labs.tu-berlin.de\ }



\date{}

\maketitle

\sloppy


\begin{abstract}
Traceroute measurements are one of our main instruments to shed
light onto the structure and properties of today's complex networks
such as the Internet. This paper studies the feasibility and
infeasibility of inferring the network topology given traceroute
data from a worst-case perspective, i.e., without any probabilistic
assumptions on, e.g., the nodes' degree distribution. We attend to a
scenario where some of the routers are anonymous, and propose two
fundamental axioms that model two basic assumptions on the
traceroute data: (1) each trace corresponds to a real path in the
network, and (2) the routing paths are at most a factor $1/\alpha$
off the shortest paths, for some parameter $\alpha\in (0,1]$. In
contrast to existing literature that focuses on the cardinality of
the set of (often only minimal) inferrable topologies, we argue that
a large number of possible topologies alone is often unproblematic,
as long as the networks have a similar structure. We hence seek to
characterize the set of topologies inferred with our axioms. We
introduce the notion of star graphs whose colorings capture the
differences among inferred topologies; it also allows us to
construct inferred topologies explicitly. We find that in general,
inferrable topologies can differ significantly in many important
aspects, such as the nodes' distances or the number of triangles.
These negative results are complemented by a discussion of a
scenario where the trace set is best possible, i.e., ``complete''.
It turns out that while some properties such as the node degrees are
still hard to measure, a complete trace set can help to determine
global properties such as the connectivity.
\end{abstract}

\clearpage

\clearpage

\section{Introduction}

Surprisingly little is known about the structure of many important
complex networks such as the Internet. One reason is the inherent
difficulty of performing accurate, large-scale and preferably
synchronous measurements from a large number of different vantage
points. Another reason are privacy and information hiding issues:
for example, network providers may seek to hide the details of their
infrastructure to avoid tailored attacks.

Since knowledge of the network characteristics is crucial for many
applications (e.g., \emph{RMTP}~\cite{rmtp}, or
\emph{PaDIS}~\cite{padis}), the research community implements
measurement tools to analyze at least the main properties of the
network. The results can then, e.g., be used to design more
efficient network protocols in the future.

This paper focuses on the most basic characteristic of the network:
its \emph{topology}. The classic tool to study topological
properties is \emph{traceroute}. Traceroute allows us to collect
traces from a given source node to a set of specified destination
nodes. A trace between two nodes contains a sequence of identifiers
describing the route traveled by the packet. However, not every node
along such a path is configured to answer with its identifier.
Rather, some nodes may be \emph{anonymous} in the sense that they
appear as stars (`$*$') in a trace. Anonymous nodes exacerbate the
exploration of a topology because already a small number of
anonymous nodes may increase the spectrum of inferrable topologies
that correspond to a trace set $\mathcal{T}$.

This paper is motivated by the observation that the mere number of
inferrable topologies alone does not contradict the usefulness or
feasibility of topology inference; if the set of inferrable
topologies is homogeneous in the sense that that the different
topologies share many important properties, the generation of all
possible graphs can be avoided: an arbitrary representative may
characterize the underlying network accurately. Therefore, we
identify important topological metrics such as diameter or maximal
node degree and examine how ``close'' the possible inferred
topologies are with respect to these metrics.

\subsection{Related Work}

Arguably one of the most influential measurement studies on the
Internet topology was conducted by the Faloutsos
brothers~\cite{faloutsos} who show that the Internet exhibits a
skewed structure: the nodes' out-degree follows a power-law
distribution. Moreover, this property seems to be invariant over
time. These results complement discoveries of similar distributions
of communication traffic which is often self-similar, and of the
topologies of natural networks such as human respiratory systems.
This property allows us to give good predictions not only on node
degree distributions but also, e.g., on the expected number of nodes
at a given hop-distance. Since~\cite{faloutsos} was published, many
additional results have been obtained, e.g., by conducting a
distributed computing approach to increase the number of measurement
points~\cite{shavitt-science}. However, our understanding remains
preliminary, and the topic continues to attract much attention from
the scientific communities. In contrast to these measurement
studies, we pursue a more formal approach, and a complete review of
the empirical results obtained over the last years is beyond the
scope of this paper.

In the field of \emph{network tomography}, topologies are explored
using pairwise end-to-end measurements, without the cooperation of
nodes along these paths. This approach is quite flexible and
applicable in various contexts, e.g., in social
networks~\cite{sigmetrics11}. For a good discussion of this approach
as well as results for a routing model along shortest and second
shortest paths see~\cite{sigmetrics11}. For
example,~\cite{sigmetrics11} shows that for sparse random graphs, a
relatively small number of cooperating participants is sufficient to
discover a network fairly well.

The classic tool to discover Internet topologies is
traceroute~\cite{traceroutedata}. Unfortunately, there are several
problems with this approach that render topology inference
difficult, such as \emph{aliasing} or \emph{load-balancing}, which
has motivated researchers to develop new tools such as \emph{Paris
Traceroute}~\cite{paristraceroute,jsac06}. Another complication
stems from the fact that routers may appear as stars in the trace
since they are overloaded or since they are configured not to send
out any ICMP responses. The lack of complete information in the
trace set renders the accurate characterization of Internet
topologies difficult.

This paper attends to the problem of anonymous nodes and assumes a
conservative, ``worst-case'' perspective that does not rely on any
assumptions on the underlying network. There are already several
works on the subject. Yao et al.~\cite{infocom03} initiated the
study of possible candidate topologies for a given trace set and
suggested computing the \emph{minimal topology}, that is, the
topology with the minimal number of anonymous nodes, which turns out
to be NP-hard. Consequently, different heuristics have been
proposed~\cite{gunes,jsac06}.

Our work is motivated by a series of papers by Acharya and Gouda.
In~\cite{sss09}, a network tracing theory model is introduced where
nodes are ``irregular'' in the sense that each node appears in at
least one trace with its real identifier. In~\cite{icdcn10},
hardness results are derived for this model. However, as pointed out
by the authors themselves, the irregular node model---where nodes
are anonymous due to high loads---is less relevant in practice and
hence they consider strictly anonymous nodes in their follow-up
studies~\cite{icdcn11}. As proved in~\cite{icdcn11}, the problem is
still hard (in the sense that there are many minimal networks
corresponding to a trace set), even with only two anonymous nodes,
symmetric routing and without aliasing.

In contrast to this line of research on cardinalities, we are
interested in the \emph{network properties}. If the inferred
topologies share the most important characteristics, the negative
results in~\cite{icdcn10,icdcn11} may be of little concern.
Moreover, we believe that a study limited to minimal topologies only
may miss important redundancy aspects of the Internet.
Unlike~\cite{icdcn10,icdcn11}, our work is constructive in the sense
that algorithms can be derived to compute inferred topologies.



\subsection{Our Contribution}

This paper initiates the study and characterization of topologies
that can be inferred from a given trace set computed with the
traceroute tool. While existing literature assuming a worst-case
perspective has mainly focused on the cardinality of minimal
topologies, we go one step further and examine specific topological
graph properties.

We introduce a formal theory of topology inference by proposing
basic axioms (i.e., assumptions on the trace set) that are used to
guide the inference process. We present a novel and we believe
appealing definition for the isomorphism of inferred topologies
which is aware of traffic paths; it is motivated by the observation
that although two topologies look equivalent up to a renaming of
anonymous nodes, the same trace set may result in different paths.
Moreover, we initiate the study of two extremes: in the first
scenario, we only require that each link appears at least once in
the trace set; interestingly, however, it turns out that this is
often not sufficient, and we propose a ``best case'' scenario where
the trace set is, in some sense, \emph{complete}: it contains paths
between all pairs of nodes.

The main result of the paper is a negative one. It is shown that
already a small number of anonymous nodes in the network renders
topology inference difficult. In particular, we prove that in
general, the possible inferrable topologies differ in many crucial
aspects.

We introduce the concept of the \emph{star graph} of a trace set
that is useful for the characterization of inferred topologies. In
particular, colorings of the star graphs allow us to constructively
derive inferred topologies. (Although the general problem of
computing the set of inferrable topologies is related to NP-hard
problems such as \emph{minimal graph coloring} and \emph{graph
isomorphism}, some important instances of inferrable topologies can
be computed efficiently.) The minimal coloring (i.e., the chromatic
number) of the star graph defines a lower bound on the number of
anonymous nodes from which the stars in the traces could originate
from. And the number of possible colorings of the star graph---a
function of the \emph{chromatic polynomial} of the star
graph---gives an upper bound on the number of inferrable topologies.
We show that this bound is tight in the sense that there are
situation where there indeed exist so many inferrable topologies.
Especially, there are problem instances where the cardinality of the
set of inferrable topologies equals the \emph{Bell number}. This
insight complements (and generalizes to arbitrary, not only minimal,
inferrable topologies) existing cardinality results.

Finally, we examine the scenario of \emph{fully explored networks}
for which ``complete'' trace sets are available. As expected,
inferrable topologies are more homogenous and can be characterized
well with respect to many properties such as node distances.
However, we also find that other properties are inherently difficult
to estimate. Interestingly, our results indicate that full
exploration is often useful for global properties (such as
connectivity) while it does not help much for more local properties
(such as node degree).

\subsection{Organization}

The remainder of this paper is organized as follows. Our theory of
topology inference is introduced in Section~\ref{sec:model}. The
main contribution is presented in Sections~\ref{sec:inferrable}
and~\ref{sec:fullyexp} where we derive bounds for general trace sets
and fully explored networks, respectively. In
Section~\ref{sec:conclusion}, the paper concludes with a discussion
of our results and directions for future research. Due to space
constraints, some proofs are moved to the appendix.

%

\section{Model}\label{sec:model}

Let $\mathcal{T}$ denote the set of traces obtained from probing
(e.g., by traceroute) a (not necessarily connected and undirected)
network $G_0=(V_0,E_0)$ with \emph{nodes} or \emph{vertices} $V_0$
(the set of routers) and \emph{links} or \emph{edges} $E_0$. We
assume that $G_0$ is static during the probing time (or that probing
is instantaneous). Each trace $T(u,v)\in\mathcal{T}$ describes a
path connecting two nodes $u,v\in V_0$; when $u$ and $v$ do not
matter or are clear from the context, we simply write $T$. Moreover,
let $d_T(u,v)$ denote the distance (number of hops) between two
nodes $u$ and $v$ in trace $T$. We define $d_{G_0}(u,v)$ to be the
corresponding shortest path distance in $G_0$. Note that a trace
between two nodes $u$ and $v$ may not describe the shortest path
between $u$ and $v$ in $G_0$.

The nodes in $V_0$ fall into two categories: \emph{anonymous} nodes
and \emph{non-anonymous} (or shorter: \emph{named}) nodes.
Therefore, each trace $T\in \mathcal{T}$ describes a sequence of
symbols representing anonymous and non-anonymous nodes. We make the
natural assumption that the first and the last node in each trace
$T$ is non-anonymous. Moreover, we assume that traces are given in a
form where non-anonymous nodes appear with a unique, anti-aliased
identifier (i.e., the multiple IP addresses corresponding to
different interfaces of a node are resolved to one identifier); an
anonymous node is represented as $*$ (``star'') in the traces. For
our formal analysis, we assign to each star in a trace set
$\mathcal{T}$ a unique identifier $i$: $*_i$. (Note that except for
the numbering of the stars, we allow identical copies of $T$ in
$\mathcal{T}$, and we do not make any assumptions on the
implications of identical traces: they may or may not describe the
same paths.) Thus, a trace $T\in \mathcal{T}$ is a sequence of
symbols taken from an alphabet $\Sigma=\mathcal{ID}\cup
\left(\bigcup_i *_i\right)$, where $\mathcal{ID}$ is the set of
non-anonymous node identifiers (IDs): $\Sigma$ is the union of the
(anti-aliased) non-anonymous nodes and the set of all stars (with
their unique identifiers) appearing in a trace set. The main
challenge in topology inference is to determine which stars in the
traces may originate from which anonymous nodes.

Henceforth, let $n=|\mathcal{ID}|$ denote the number of
non-anonymous nodes and let $s=\left|\bigcup_i
*_i\right|$ be the number of stars in $\mathcal{T}$; similarly, let
$a$ denote the number of anonymous nodes in a topology.
Let $N=n+s=|\Sigma|$ be the total number of symbols
occurring in $\mathcal{T}$.

Clearly, the process of topology inference depends on the
assumptions on the measurements. In the following, we postulate the
fundamental axioms that guide the reconstruction. First, we make the
assumption that each link of $G_0$ is visited by the measurement
process, i.e., it appears as a transition in the trace set
$\mathcal{T}$. In other words, we are only interested in inferring
the (sub-)graph for which measurement data is available.
\begin{shaded}
\RULEZERO\ (\emph{Complete Cover}): Each edge of $G_0$ appears at
least once in some trace in $\mathcal{T}$.
\end{shaded}

The next fundamental axiom assumes that traces always represent
paths on $G_0$.
\begin{shaded}
\RULEONE\ (\emph{Reality Sampling}): For every trace $T\in
\mathcal{T}$, if the distance between two symbols $\sigma_1,\sigma_2
\in T$ is $d_T(\sigma_1,\sigma_2)=k$, then there exists a path
(i.e., a walk without cycles) of length $k$ connecting two (named or
anonymous) nodes $\sigma_1$ and $\sigma_2$ in $G_0$.
\end{shaded}
The following axiom captures the consistency of the routing protocol
on which the traceroute probing relies. In the current Internet,
policy routing is known to have in impact both on the route
length~\cite{tangmunarunkit2002impact} and on the convergence
time~\cite{impact-roger}.
\begin{shaded}
\RULETWO\ (\emph{$\alpha$-(Routing) Consistency}): There exists an
$\alpha \in (0,1]$ such that, for every trace $T\in \mathcal{T}$, if
$d_T(\sigma_1,\sigma_2)=k$ for two entries $\sigma_1,\sigma_2$ in
trace $T$, then the shortest path connecting the two (named or
anonymous) nodes corresponding to $\sigma_1$ and $\sigma_2$ in $G_0$
has distance at least $\lceil \alpha k \rceil$.
\end{shaded}
Note that if $\alpha=1$, the routing is a shortest path routing.
Moreover, note that if $\alpha=0$, there can be loops in the paths,
and there are hardly any topological constraints, rendering almost
any topology inferrable. (For example, the complete graph with one
anonymous router is always a solution.)

A natural axiom to merge traces is the following.
\begin{shaded} \RULETHREE\ (\emph{Trace Merging}): For two
traces $T_1,T_2 \in \mathcal{T}$ for which $\exists
\sigma_1,\sigma_2,\sigma_3$, where $\sigma_2$ refers to a named
node, such that $d_{T_1}(\sigma_1,\sigma_2)=i$ and
$d_{T_2}(\sigma_2,\sigma_3)=j$, it holds that the distance between
two nodes $u$ and $v$ corresponding to $\sigma_1$ and $\sigma_2$,
respectively, in $G_0$, is at most $d_{G_0}(\sigma_1,\sigma_3)\leq
i+j$.
\end{shaded}

Any topology $G$ which is consistent with these axioms (when applied
to $\mathcal{T}$) is called \emph{inferrable} from $\mathcal{T}$.
\begin{definition}[Inferrable Topologies]
A topology $G$ is ($\alpha$-consistently) \emph{inferrable} from a
trace set $\mathcal{T}$ if axioms \RULEZERO, \RULEONE, \RULETWO\
(with parameter $\alpha$), and \RULETHREE\ are fulfilled.
\end{definition}

We will refer by $\mathcal{G}_{\mathcal{T}}$ to the set of
topologies inferrable from $\mathcal{T}$. Please note the following
important observation.

\begin{remark}
While we generally have that $G_0 \in \mathcal{G}_{\mathcal{T}}$,
since $\mathcal{T}$ was generated from $G_0$ and \RULEZERO,
\RULEONE, \RULETWO\, and \RULETHREE\ are fulfilled by definition,
there can be situations where an $\alpha$-consistent trace set for
$G_0$ contradicts \RULEZERO: some edges may not appear in
$\mathcal{T}$. If this is the case, we will focus on the inferrable
topologies containing the links we know, even if $G_0$ may have
additional, hidden links that cannot be explored due to the high
$\alpha$ value.
\end{remark}

The main objective of a topology inference algorithm $\ALG$ is to
compute topologies which are consistent with these axioms.
Concretely, $\ALG$'s input is the trace set $\mathcal{T}$ together
with the parameter $\alpha$ specifying the assumed routing
consistency. Essentially, the goal of any topology inference
algorithm $\ALG$ is to compute a mapping of the symbols $\Sigma$
(appearing in $\mathcal{T}$) to nodes in an inferred topology $G$;
or, in case the input parameters $\alpha$ and $\mathcal{T}$ are
contradictory, reject the input. This mapping of symbols to nodes
implicitly describes the edge set of $G$ as well: the edge set is
unique as all the transitions of the traces in $\mathcal{T}$ are now
unambiguously tied to two nodes.

\begin{wrapfigure}{r}{0.42\textwidth}
        \vspace*{-.4cm}
    \includegraphics[width=0.38\textwidth]{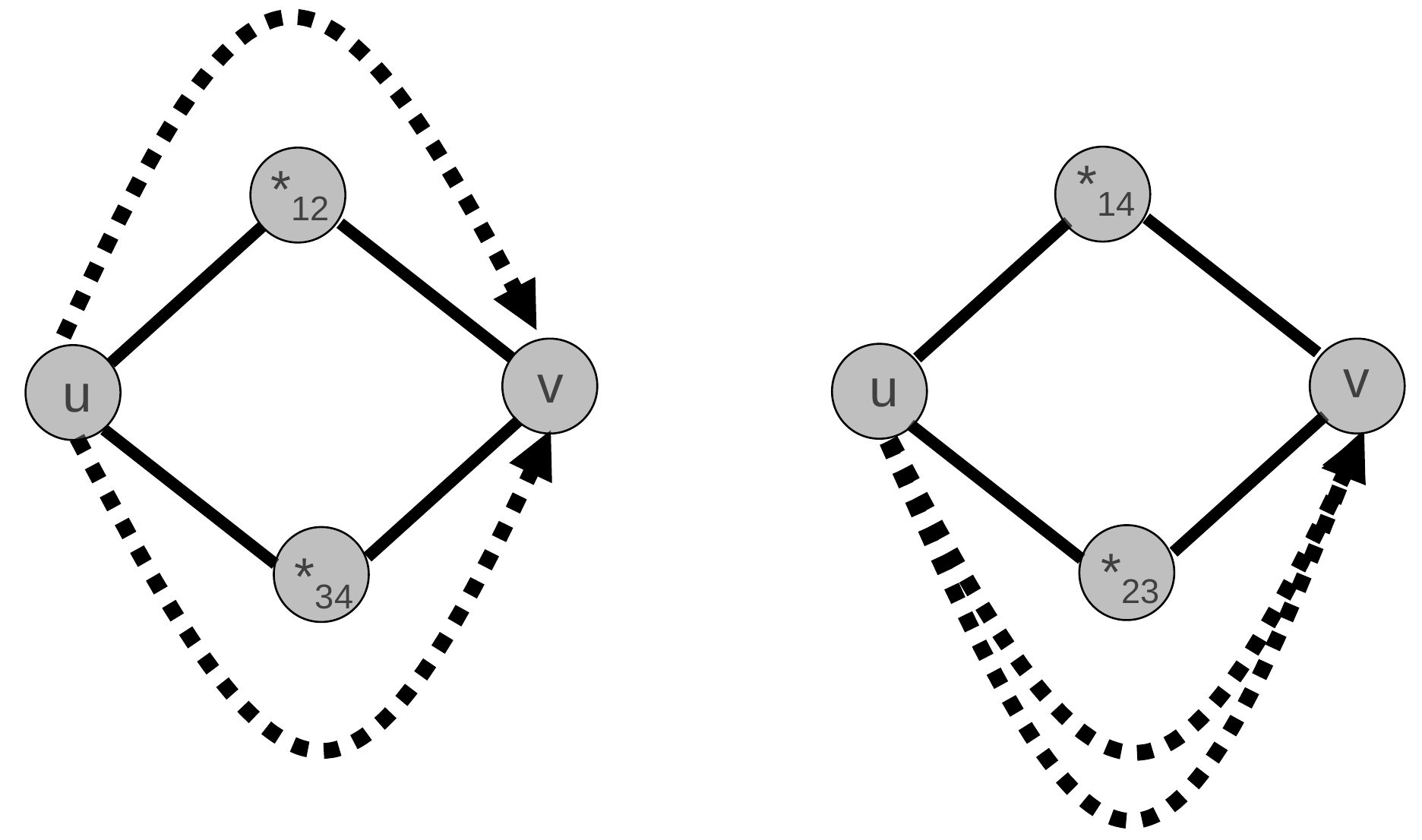}
    \caption{\bf Two non-isomorphic inferred
topologies, i.e., different mapping functions lead to these
topologies.}\label{fig:isomorph}
    \vspace*{-0.2cm}
\end{wrapfigure}

So far, we have ignored an important and non-trivial question: When
are two topologies $G_1,G_2\in \mathcal{G}_{\mathcal{T}}$ different
(and hence appear as two independent topologies in
$\mathcal{G}_{\mathcal{T}}$)? In this paper, we pursue the following
approach: We are not interested in purely topological isomorphisms,
but we care about the identifiers of the non-anonymous nodes, i.e.,
we are interested in the locations of the non-anonymous nodes and
their distance to other nodes. For anonymous nodes, the situation is
slightly more complicated: one might think that as the nodes are
anonymous, their ``names'' do not matter. Consider however the
example in Figure~\ref{fig:isomorph}: the two inferrable topologies
have two anonymous nodes, once where $\{*_1,*_2\}$ plus
$\{*_3,*_4\}$ are merged into one node each in the inferrable
topology and once where $\{*_1,*_4\}$ plus $\{*_2,*_3\}$ are merged
into one node each in the inferrable topology. In this paper, we
regard the two topologies as different, for the following reason:
Assume that there are two paths in the network, one
$u\rightsquigarrow
*_2 \rightsquigarrow v$ (e.g., during day time) and one
$u\rightsquigarrow *_3 \rightsquigarrow v$ (e.g., at night);
clearly, this traffic has different consequences and hence we want
to be able to distinguish between the two topologies described
above. In other words, our notion of isomorphism of inferred
topologies is \emph{path-aware}.

It is convenient to introduce the following $\map$ function.
Essentially, an inference algorithm computes such a mapping.
\begin{definition}[Mapping Function $\map$]\label{def:topinference}
Let $G=(V,E) \in \mathcal{G}_{\mathcal{T}}$ be a topology inferrable
from $\mathcal{T}$. A topology inference algorithm describes a
surjective mapping function $\map: \Sigma \to V$. For the set of
non-anonymous nodes in $\Sigma$, the mapping function is bijective;
and each star is mapped to exactly one node in $V$, but multiple
stars may be assigned to the same node. Note that for any $\sigma
\in \Sigma$, $\map(\sigma)$ uniquely identifies a node $v\in V$.
More specifically, we assume that $\map$ assigns labels to the nodes
in $V$: in case of a named node, the label is simply the node's
identifier; in case of anonymous nodes, the label is $*_{\beta}$,
where $\beta$ is the concatenation of the \emph{sorted} indices of
the stars which are merged into node $*_{\beta}$.
\end{definition}
With this definition, two topologies $G_1,G_2\in
\mathcal{G}_{\mathcal{T}}$ differ if and only if they do not
describe the identical ($\map$-) labeled topology. We will use this
$\map$ function also for $G_0$, i.e., we will write $\map(\sigma)$
to refer to a symbol $\sigma$'s corresponding node in $G_0$.

In the remainder of this paper, we will often assume that \RULEZERO\
is given. Moreover, note that \RULETHREE\ is redundant. Therefore,
in our proofs, we will not explicitly cover \RULEZERO, and it is
sufficient to show that \RULEONE\ holds to prove that \RULETHREE\ is
satisfied.
\begin{lemma}
\RULEONE\ implies \RULETHREE.
\end{lemma}
\begin{proof}
Let $\mathcal{T}$ be a trace set, and $G
\in\mathcal{G}_{\mathcal{T}}$. Let $\sigma_1,\sigma_2,\sigma_3 $
s.t.~$\exists T_1,T_2 \in \mathcal{T}$ with $\sigma_1 \in T_1,
\sigma_3\in T_2$ and $\sigma_2 \in T_1 \cap T_2$. Let
$i=d_{T_1}(\sigma_1,\sigma_2)$ and $j=d_{T_2}(\sigma_1,\sigma_3)$.
Since any inferrable topology $G$ fulfills \RULEONE, there is a path
$\pi_1$ of length at most $i$ between the nodes corresponding to
$\sigma_1$ and $\sigma_2$ in $G$ and a path $\pi_2$ of length at
most $j$ between the nodes corresponding to $\sigma_2$ and
$\sigma_3$ in $G$. The combined path can only be shorter, and hence
the claim follows.
\end{proof}


\section{Inferrable Topologies}\label{sec:inferrable}

What insights can be obtained from topology inference with minimal
assumptions, i.e., with our axioms? Or what is the structure of the
inferrable topology set $\mathcal{G}_{\mathcal{T}}$? We first make
some general observations and then examine different graph metrics
in more detail.

\subsection{Basic Observations}

Although the generation of the entire topology set
$\mathcal{G}_{\mathcal{T}}$ may be computationally hard, some
instances of $\mathcal{G}_{\mathcal{T}}$ can be computed
efficiently. The simplest possible inferrable topology is the
so-called \emph{canonic graph} $G_C$: the topology which assumes
that all stars in the traces refer to different anonymous nodes. In
other words, if a trace set $\mathcal{T}$ contains
$n=|\mathcal{ID}|$ named nodes and $s$ stars, $G_C$ will contain
$|V(G_C)|=N=n+s$ nodes.

\begin{definition}[Canonic Graph $G_C$]\label{def:canonic}
  The \emph{canonic graph} is defined by $G_C(V_C,E_C)$ where
  $V_C=\Sigma$ is the set of (anti-aliased) nodes appearing in $\mathcal{T}$
  (where each star is considered a unique anonymous node)
  and where $\{\sigma_1,\sigma_2\} \in E_C \Leftrightarrow
  \exists T \in \mathcal{T} ,  T=(\ldots,\sigma_1,
  \sigma_2,\ldots)$, i.e., $\sigma_1$ follows after $\sigma_2$
  in some trace $T$ ($\sigma_1,\sigma_2\in T$ can be either
  non-anonymous nodes or stars). Let $d_C(\sigma_1,\sigma_2)$ denote
  the \emph{canonic distance} between two nodes, i.e., the length of a
  shortest path in $G_C$ between the nodes $\sigma_1$ and $\sigma_2$.
\end{definition}

Note that $G_C$ is indeed an inferrable topology. In this case,
$\map: \Sigma \rightarrow \Sigma$ is the identity function. The
proof appears in the appendix.
\begin{theorem}\label{thm:gc-inferable}
  $G_C$ is inferrable from $\mathcal{T}$.
\end{theorem}
$G_C$ can be computed efficiently from $\mathcal{T}$: represent each
non-anonymous node and star as a separate node, and for any pair of
consecutive entries (i.e., nodes) in a trace, add the corresponding
link. The time complexity of this construction is linear in the size
of $\mathcal{T}$.

\cancel{ Next we show that the constraints on the routing paths
gives a necessary condition when two stars in different traces
cannot represent the same node in $G_0$.
\begin{lemma}\label{eqn:lemmaInference}
Assume any $\alpha>0$. If there is a trace $T\in\mathcal{T}$
containing two stars $*_1,*_2\in T$, or if there are two traces
$\exists T_1, T_2 \in \mathcal{T}$ with $d_{T_1}(v,*_1)<\lceil
\alpha\cdot d_{T_2}(v,*_2)\rceil$ for some non-anonymous node $v\in
V_0$, then $\map(*_1) \in G_0 \neq \map(*_2) \in G_0$. If $T_1 =
T_2$, $\map(*_1)\neq \map(*_2)$.
\end{lemma}
\begin{proof}
If the stars occur in the same trace, the claim follows trivially
due to the loop-free routing ($\alpha>0$). The second proof is by
contradiction. Assume $\map(*_1)= \map(*_2)$ represent the same node
$w$ of $G_0$. Then we know from \RULETWO\ that
$d_{G_0}(v,w)\geq\lceil \alpha\cdot d_{T_2}(v,w) \rceil$ and by
construction (and \RULEONE) $d_{G_0}(v,w)\leq d_{T_1}(v,w) $, which
contradicts $d_{T_1}(v,w)<\lceil \alpha\cdot \dist(v,w)\rceil$.
\end{proof}
}


With the definition of the canonic graph, we can derive the
following lemma which establishes a necessary condition when two
stars cannot represent the same node in $G_0$ from constraints on
the routing paths. This is useful for the characterization of
inferred topologies.
\begin{lemma}\label{lem:lemmaInferencemulti}
Let $*_1,*_2$ be two stars occurring in some traces in
$\mathcal{T}$. $*_1,*_2$ cannot be mapped to the same node, i.e.,
$\map(*_1) \neq \map(*_2)$, without violating the axioms in the
following conflict situations:
\begin{itemize}
\item[(i)] if $*_1\in T_1$ and $*_2\in T_2$, and $T_1$
describes a too long path between anonymous node $\map(*_1)$ and
non-anonymous node $u$, i.e., $ \lceil \alpha \cdot
d_{T_1}(*_1,u)\rceil
> d_{C}(u,*_2)$.
\item[(ii)] if $*_1\in T_1$ and $*_2\in T_2$, and there exists a trace $T$ that contains a path between two non-anonymous nodes $u$
and $v$ and $\lceil \alpha \cdot d_{T}(u,v)\rceil >
d_{C}(u,*_1)+d_{C}(v,*_2). $
\end{itemize}
\end{lemma}
\begin{proof}
The first proof is by contradiction. Assume $\map(*_1)= \map(*_2)$
represents the same node $v$ of $G_0$, and that $\lceil \alpha \cdot
d_{T_1}(v,u)\rceil > d_{C}(u,v)$. Then we know from \RULETWO\ that
$d_{C}(v,u)\geq d_{G_0}(v,u)\geq\lceil \alpha \cdot d_{T_1}(u,v)
\rceil
> d_{C}(v,u)$, which yields the desired contradiction.

Similarly for the second proof. Assume for the sake of contradiction
that $\map(*_1)= \map(*_2)$ represents the same node $w$ of $G_0$,
and that $\lceil \alpha \cdot d_{T}(u,v)\rceil
> d_{C}(u,w)+d_{C}(v,w)$. Due to the triangle inequality, we have that
$d_{C}(u,w)+d_{C}(v,w)\geq d_{C}(u,v)$ and hence, $\lceil \alpha
\cdot d_{T}(u,v)\rceil > d_{C}(u,v)$, which contradicts the fact
that $G_C$ is inferrable (Theorem~\ref{thm:gc-inferable}).
\end{proof}

Lemma~\ref{lem:lemmaInferencemulti} can be applied to show that a
topology is not inferrable from a given trace set because it merges
(i.e., maps to the same node) two stars in a manner that violates
the axioms. Let us introduce a useful concept for our analysis: the
\emph{star graph} that describes the conflicts between stars.
\begin{definition}[Star Graph $G_*$]\label{def:stargraph}
The \emph{star graph} $G_*(V_*,E_*)$ consists of vertices $V_*$
representing stars in traces, i.e., $V_*=\bigcup_i *_i$. Two
vertices are connected if and only if they must differ according to
Lemma~\ref{lem:lemmaInferencemulti}, i.e., $\{*_1,*_2\}\in E_*$ if
and only if at least one of the conditions of
Lemma~\ref{lem:lemmaInferencemulti} hold for $*_1,
*_2$.
\end{definition}

Note that the star graph $G_*$ is unique and can be computed
efficiently for a given trace set $\mathcal{T}$: Conditions~(i)
and~(ii) can be checked by computing $G_C$. However, note that while
$G_*$ specifies some stars which cannot be merged, the construction
is not sufficient: as Lemma~\ref{lem:lemmaInferencemulti} is based
on $G_C$, additional links might be needed to characterize the set
of inferrable and $\alpha$-consistent topologies
$\mathcal{G}_{\mathcal{T}}$ exactly. In other words, a topology $G$
obtained by merging stars that are adjacent in $G_*$ is never
inferrable ($G\not\in\mathcal{G}_{\mathcal{T}}$); however, merging
non-adjacent stars does not guarantee that the resulting topology is
inferrable.

What do star graphs look like? The answer is \emph{arbitrarily}: the
following lemma states that the set of possible star graphs is
equivalent to the class of general graphs. This claim holds for any
$\alpha$. The proof appears in the appendix.
\begin{lemma}\label{lemma:allstars}
For any graph $G=(V,E)$, there exists a trace set $\mathcal{T}$ such
that $G$ is the star graph for $\mathcal{T}$.
\end{lemma}

The problem of computing inferrable topologies is related to the
vertex colorings of the star graphs. We will use the following
definition which relates a vertex coloring of $G_*$ to an inferrable
topology $G$ by contracting independent stars in $G_*$ to become one
anonymous node in $G$. For example, observe that a maximum coloring
treating every star in the trace as a separate anonymous node
describes the inferrable topology $G_C$.
\begin{definition}[Coloring-Induced Graph]
\label{def:colorinducedGraph}
  Let $\gamma$ denote a coloring of $G_*$ which assigns colors $1,\ldots,k$ to the vertices of $G_*$:
  $\gamma: V_* \rightarrow \{1,\ldots,k\}$.
 We require that $\gamma$ is a proper coloring of $G_*$, i.e., that
 different anonymous nodes are assigned different colors:
 $\{u,v\} \in E_* \Rightarrow
 \gamma(u) \neq \gamma(v)$.
 $G_{\gamma}$ is defined
 as the topology \emph{induced} by $\gamma$. $G_{\gamma}$ describes the graph $G_C$ where
 nodes of the same color are contracted: two vertices $u$ and $v$ represent the same
 node in $G_{\gamma}$, i.e., $\map(*_i)=\map(*_j)$, if and only if $\gamma(*_i)=\gamma(*_j)$.
\end{definition}


The following two lemmas establish an intriguing relationship
between colorings of $G_*$ and inferrable topologies. Also note that
Definition~\ref{def:colorinducedGraph} implies that two different
colorings of $G_*$ define two non-isomorphic inferrable topologies.

We first show that while a coloring-induced topology always fulfills
\RULEONE, the routing consistency is sacrificed. The proof appears
in the appendix.
\begin{lemma}\label{thm:bijection}
Let $\gamma$ be a proper coloring of $G_*$. The coloring induced
topology $G_{\gamma}$ is a topology fulfilling \RULETWO\ with a
routing consistency of $\alpha'$, for some positive $\alpha'$.
\end{lemma}

An inferrable topology always defines a proper coloring on $G_*$.
\begin{lemma}\label{lemma:bijection2}
  Let $\mathcal{T}$ be a trace set and $G_*$ its corresponding star graph. If a
  topology $G$ is inferrable from $\mathcal{T}$, then $G$ induces a proper coloring
  on  $G_*$.
\end{lemma}
\begin{proof}
  For any $\alpha$-consistent inferrable topology $G$ there exists some mapping function \map\
  that assigns each symbol of $\mathcal{T}$ to a corresponding node in $G$ (cf~Definition~\ref{def:topinference}), and this mapping function
  gives a coloring on $G_*$ (i.e., merged stars appear as nodes of the same
  color in $G_*$). The coloring must be proper: due to Lemma~\ref{lem:lemmaInferencemulti}, an inferrable
  topology can never merge adjacent nodes of $G_*$.
\end{proof}



The colorings of $G_*$ allow us to derive an upper bound on the
cardinality of $\mathcal{G}_{\mathcal{T}}$.
\begin{theorem}\label{thm:upper}
Given a trace set $\mathcal{T}$ sampled from a network $G_0$ and
$\mathcal{G}_{\mathcal{T}}$, the set of topologies inferrable from
$\mathcal{T}$, it holds that:
$$\displaystyle \sum_{k=\gamma(G_*)}^{\vert V_* \vert}
P(G_*,k)/k! \geq \vert \mathcal{G}_{\mathcal{T}} \vert ,
$$
\noindent where $\gamma(G_*)$ is the chromatic number of $G_*$ and
$P(G_*,k)$ is the number of colorings of $G_*$ with $k$ colors
(known as the \emph{chromatic polynomial} of $G_*$).
\end{theorem}
\begin{proof}
The proof follows directly from Lemma~\ref{lemma:bijection2} which
shows that each inferred topology has proper colorings, and the fact
that a coloring of $G_*$ cannot result in two different inferred
topologies, as the coloring uniquely describes which stars to merge
(Lemma~\ref{thm:bijection}). In order to account for isomorphic
colorings, we need to divide by the number of color permutations.
%
%
\end{proof}

Note that the fact that $G_*$ can be an arbitrary graph
(Lemma~\ref{lemma:allstars}) implies that we cannot exploit some
special properties of $G_*$ to compute colorings of $G_*$ and
$\gamma(G_*)$. Also note that the exact computation of the upper
bound is hard, since the minimal coloring as well as the chromatic
polynomial of $G_*$ (in P$\sharp$) is needed. To complement the
upper bound, we note that star graphs with a small number of
conflict edges can indeed result in a large number of inferred
topologies.
\begin{theorem}\label{thm:lower}
For any $\alpha>0$, there is a trace set for which the number of
non-isomorphic colorings of $G_*$ equals
$|\mathcal{G}_{\mathcal{T}}|$, in particular
$|\mathcal{G}_{\mathcal{T}}|= B_s$, where
$\mathcal{G}_{\mathcal{T}}$ is the set of inferrable and
$\alpha$-consistent topologies, $s$ is the number of stars in
$\mathcal{T}$, and $B_s$ is the \emph{Bell number}
of $s$. Such a
trace set can originate from a $G_0$ network with one anonymous node
only.
\end{theorem}
\begin{proof}
Consider a trace set
$\mathcal{T}=\{(\sigma_i,*_i,\sigma'_i)_{i=1,\ldots,s}\}$ (e.g.,
obtained from exploring a topology $G_0$ where one anonymous center
node is connected to $2s$ named nodes). The trace set does not
impose any constraints on how the stars relate to each other, and
hence, $G_*$ does not contain any edges at all; even when stars are
merged, there are no constraints on how the stars relate to each
other. Therefore, the star graph for $\mathcal{T}$ has
$B_s=\sum_{j=0}^s S_{(s,j)}$ colorings, where $S_{(s,j)}=1/j!\cdot
\sum_{\ell=0}^j (-1)^{\ell}\binom{j}{\ell}(j-\ell)^s$ is the number
of ways to group $s$ nodes into $j$ different, disjoint non-empty
subsets (known as the \emph{Stirling number of the second kind}).
Each of these colorings also describes a distinct inferrable
topology as $\map$ assigns unique labels to anonymous nodes stemming
from merging a group of stars (cf
Definition~\ref{def:topinference}).
\end{proof}

\subsection{Properties}

Even if the number of inferrable topologies is large, topology
inference can still be useful if one is mainly interested in the
properties of $G_0$ and if the ensemble $\mathcal{G}_{\mathcal{T}}$
is homogenous with respect to these properties; for example, if
``most'' of the instances in $\mathcal{G}_{\mathcal{T}}$ are close
to $G_0$, there may be an option to conduct an efficient sampling
analysis on random representatives. Therefore, in the following, we
will take a closer look how much the members of
$\mathcal{G}_{\mathcal{T}}$ differ.

Important metrics to characterize inferrable topologies are, for
instance, the graph size, the diameter $\diam(\cdot)$, the number of
triangles $C_{3}(\cdot)$ of $G$, and so on. In the following, let
$G_1=(V_1,E_1),G_2=(V_2,E_2) \in \mathcal{G}_{\mathcal{T}}$ be two
arbitrary representatives of $\mathcal{G}_{\mathcal{T}}$.

As one might expect, the graph size can be estimated quite well.
\begin{lemma}\label{lemma:nofnodesedges}
It holds that $|V_1|-|V_2|\leq s - \gamma(G_*)\leq  s -1$ and
$|V_1|/|V_2|\leq (n+s)/(n+\gamma(G_*))\leq (2+s)/3$. Moreover,
$|E_1|-|E_2|\leq 2(s-\gamma(G_*))$ and $|E_1|/|E_2|\leq
(\nu+2s)/(\nu+2)\leq s$, where $\nu$ denotes the number of edges
between non-anonymous nodes. There are traces with inferrable
topology $G_1, G_2$ reaching these bounds.
\end{lemma}

Observe that inferrable topologies can also differ in the number of
connected components. This implies that the shortest distance
between two named nodes can differ arbitrarily between two
representatives in $\mathcal{G}_{\mathcal{T}}$.
\begin{lemma}\label{lemma:components}
  Let $\Comp(G)$ denote the number of connected components of a topology $G$.
  Then, $|\Comp(G_1)-\Comp(G_2)|\leq n/2$. There are instances $G_1,G_2$
that reach this bound.
\end{lemma}
\begin{proof}
  Consider the trace set $\mathcal{T}=\{T_i, i=1 \ldots \lfloor n/2
  \rfloor\}$ in which $T_i=\{n_{2i},*_{i},n_{2i+1}\}$. Since $i\neq j
  \Rightarrow T_i \cap T_j = \emptyset$, we have $\vert E_* \vert =0$. Take
  $G_1$ as the $1$-coloring of $G_*$: $G_1$ is a topology with one anonymous
  node connected to all named nodes. Take $G_2$ as the $\lfloor n/2 \rfloor$-coloring
  of the star graph: $G_2$ has $\lfloor n/2 \rfloor$ distinct
  connected components (consisting of three nodes).

  \emph{Upper bound:} For the sake of contradiction, suppose $\exists \mathcal{T}$ s.t.~$\vert
  \Comp(G_1) - \Comp(G_2) \vert > \lfloor n/2 \rfloor$. Let us
  assume that
  $G_1$ has the most connected components: $G_1$ has at least $ \lfloor
  n/2 \rfloor +1$ more connected components than $G_2$. Let $C$ refer to a connected component of $G_2$
  whose nodes are not connected in $G_1$. This means that $C$ contains at least one anonymous node. Thus, $C$ contains
  at least two named nodes (since a trace $T$ cannot start or end by a star).
  There must exist at least $\lfloor n/2 \rfloor +1$ such connected component
  $C$. Thus $G_2$ has to contain at least $2(\lfloor n/2 \rfloor +1)
  \geq n+1 $ named nodes. Contradiction.
\end{proof}

An important criterion for topology inference regards the distortion
of shortest paths.
\begin{definition}[Stretch]\label{def:diststretch}
The maximal ratio of the distance of two non-anonymous nodes in
$G_0$ and a connected topology $G$ is called the \emph{stretch}
$\rho$: $ \rho=\max_{u,v\in \mathcal{ID}(G_0)}
\max\{\dist_{G_0}(u,v)/\dist_{G}(u,v),\dist_{G}(u,v)/\dist_{G_0}(u,v)\}.
$
\end{definition}

From Lemma~\ref{lemma:components} we already know that inferrable
topologies can differ in the number of connected components, and
hence, the distance and the stretch between nodes can be arbitrarily
wrong. Hence, in the following, we will focus on connected graphs
only. However, even if two nodes are connected, their distance can
be much longer or shorter than in $G_0$.
Figure~\ref{fig:consistency} gives an example. Both topologies are
inferrable from the traces $T_1=(v,*,v_1,\ldots,v_{k},u)$ and
$T_2=(w,*,w_1,\ldots,w_{k},u)$. One inferrable topology is the
canonic graph $G_C$ (Figure~\ref{fig:consistency} \emph{left}),
whereas the other topology merges the two anonymous nodes
(Figure~\ref{fig:consistency} \emph{right}). The distances between
$v$ and $w$ are $2(k+2)$ and $2$, respectively, implying a stretch
of $k+2$.
\begin{wrapfigure} {r}{0.42\textwidth}
\includegraphics[width=0.38\textwidth]{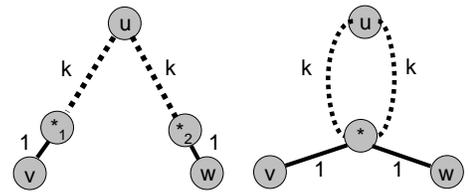}
        \vspace*{-.3cm}
\caption{\bf Due to the lack of a trace between $v$ and $w$, the
stretch of an inferred topology can be
large.}\label{fig:consistency}
        \vspace*{-.2cm}
\end{wrapfigure}

\begin{lemma}\label{lemma:distortion}
Let $u$ and $v$ be two arbitrary named nodes in the connected
topologies $G_1$ and $G_2$. Then, even for only two stars in the
trace set, it holds for the stretch that $ \rho \leq (N-1)/2$. There
are instances $G_1,G_2$ that reach this bound.
\end{lemma}

We now turn our attention to the diameter and the degree.
\begin{lemma}\label{lemma:diameter}
For connected topologies $G_1,G_2$ it holds that
$\diam(G_1)-\diam(G_2)\leq (s-1)/s\cdot \diam(G_C)\leq (s-1)(N-1)/s$
and $\diam(G_1)/\diam(G_2)\leq s$, where $\diam$ denotes the graph
diameter and $\diam(G_1)>\diam(G_2)$. There are instances $G_1,G_2$
that reach these bounds.
\end{lemma}
\begin{wrapfigure}{r}{0.5\textwidth}
        \vspace*{-.3cm}
\includegraphics[width=0.44\textwidth]{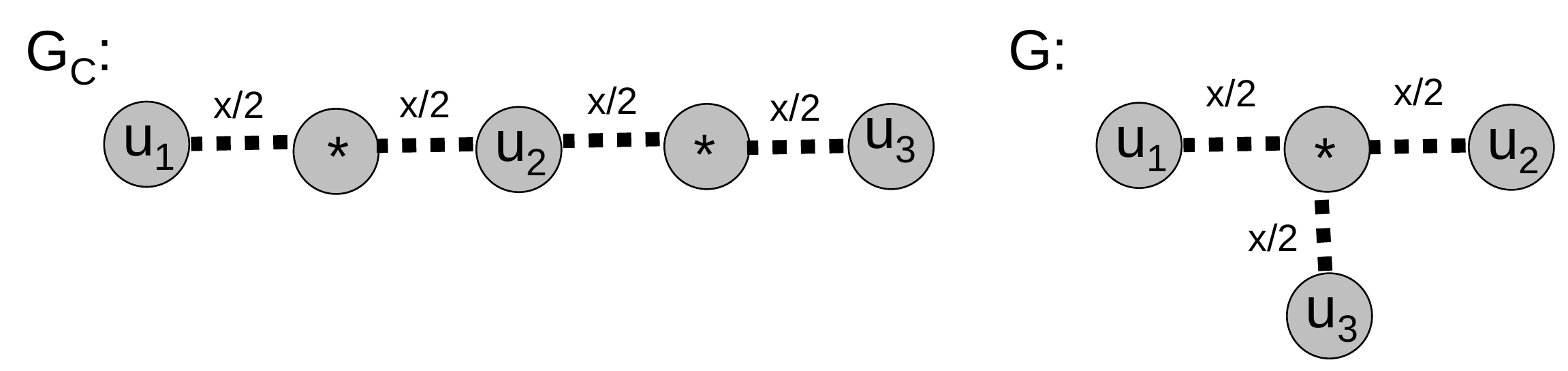}
        \vspace*{-.3cm}
 \caption{\bf Estimation error for diameter.}\label{fig:diam-a}
\end{wrapfigure}
\begin{proof}
\emph{Upper bound:} As $G_C$ does not merge any stars, it describes
the network with the largest diameter. Let $\pi$ be a longest path
between two nodes $u$ and $v$ in $G_C$. In the extreme case, $\pi$
is the only path determining the network diameter and $\pi$ contains
all star nodes. Then, the graph where all $s$ stars are merged into
one anonymous node has a minimal diameter of at least
$\diam(G_C)/s$.

\emph{Example meeting the bound:} Consider the trace set
$\mathcal{T}$ $=\{(u_1,\ldots,*_1,\ldots,u_2),$
$(u_2,\ldots,*_2,\ldots,u_3),$ $\ldots,$ $(u_{s},\ldots,*_s,$
$\ldots,u_{s+1})\}$ with $x$ named nodes and star in the middle
between $u_i$ and $u_{i+1}$ (assume $x$ to be even, $x$ does not
include $u_i$ and $u_{i+1}$ ). It holds that
$\diam(G_C)=s\cdot(x+2)$ whereas in a graph $G$ where all stars are
merged, $\diam(G)=x+2$. There are $n=s(x+1)$ non-anonymous nodes, so
$x=(n-s-1)/s$. Figure~\ref{fig:diam-a} depicts an example.
\end{proof}

\begin{lemma}\label{lemma:degree}
For the maximal node degree $\degree$, we have
$\degree(G_1)-\degree(G_2)\leq 2(s-\gamma(G_*))$ and
$\degree(G_1)/\degree(G_2)\leq s-\gamma(G_*)+1$. There are instances
$G_1,G_2$ that reach these bounds.
\end{lemma}

Another important topology measure that indicates how well meshed
the network is, is the number of triangles.
\begin{lemma}\label{lemma:noftriangles}
Let $C_3(G)$ be the number of cycles of length $3$ of the graph $G$.
It holds that $C_3(G_1)-C_3(G_2)\leq 2s(s-1)$, which can be reached.
The relative error $C_3(G_1)/C_3(G_2)$ can be arbitrarily large
unless the number of links between non-anonymous nodes exceeds
$n^2/4$ in which case the ratio is upper bounded by $2s(s-1)+1$.
\end{lemma}
\begin{proof}
\emph{Upper bound:} Each node which is part of a triangle has at
least two incident edges. Thus, a node $v$ can be part of at most
$\degree(v)\choose2$ triangles, where $\degree(v)$ denotes $v$'s
degree. As a consequence the number of triangles containing an
anonymous node in an inferrable topology with $a$ anonymous nodes
$u_1,\ldots u_a$ is at most $\sum_{j=1}^{a}$~$\degree(u_j)\choose2$.
Given $s$, this sum is maximized if $a=1$ and $\degree(u_1)=2s$ as
$2s$ is the maximum degree possible due to Lemma~\ref{lemma:degree}.
Thus there can be at most $s\cdot(2 s-1)$ triangles containing an
anonymous node in $G_1$. The number of triangles with at least one
anonymous node is minimized in $G_C$ because in the canonic graph
the degrees of the anonymous nodes are minimized, i.e, they are
always exactly two. As a consequence there cannot be more than $s$
such triangles in $G_C$.

If the number of such triangles
in $G_C$ is smaller by $x$, then the number of of triangles with at
least one anonymous node in the topology $G_1$ is upper bounded by
$s\cdot(2s-1)-x$. The difference between the triangles in $G_1$ and
$G_2$ is thus at most $s(2s-1)-x-s+x=2s(s-1)$.

\emph{Example meeting this bound:} If the non-anonymous
nodes form a complete graph and all star nodes can be merged into
one node in $G_1$ and $G_2=G_C$, then the difference in the number
of triangles matches the upper bound. Consequently it holds for the
ratio of triangles with anonymous nodes that it does not exceed
$(s(2s-1)-x)/(s-x).$ Thus the ratio can be infinite, as $x$ can
reach $s$. However, if the number of links between $n$ non-anonymous
nodes exceeds $n^2/4$ then there is at least one triangle, as the
densest complete bipartite graph contains at most $n^2/4$ links.
\end{proof}

\section{Full Exploration}\label{sec:fullyexp}

So far, we assumed that the trace set $\mathcal{T}$ contains each
node and link of $G_0$ at least once. At first sight, this seems to
be the best we can hope for. However, sometimes traces exploring the
vicinity of anonymous nodes in different ways yields additional
information that help to characterize $\mathcal{G}_{\mathcal{T}}$
better.

This section introduces the concept of \emph{fully explored
networks}: $\mathcal{T}$ contains sufficiently many traces such that
the distances between non-anonymous nodes can be estimated
accurately.
\begin{definition}[Fully Explored Topologies]\label{def:FE}
A topology $G_0$ is fully explored by a trace set $\mathcal{T}$ if
it contains all nodes and links of $G_0$ and for each pair $\{u,v\}$
of non-anonymous nodes in the same component of $G_0$ there exists a
trace $T\in\mathcal{T} $ containing both nodes $u\in T$ and $v\in
T$.
\end{definition}

In some sense, a trace set for a fully explored network is the best
we can hope for. Properties that cannot be inferred well under the
fully explored topology model are infeasible to infer without
additional assumptions on $G_0$. In this sense, this section
provides upper bounds on what can be learned from topology
inference. In the following, we will constrain ourselves to routing
along shortest paths only ($\alpha=1$).



Let us again study the properties of the family of inferrable
topologies fully explored by a trace set. Obviously, all the upper
bounds from Section~\ref{sec:inferrable} are still valid for fully
explored topologies. In the following, let $G_1, G_2 \in
\mathcal{G}_{\mathcal{T}}$ be arbitrary representatives of
$\mathcal{G}_{\mathcal{T}}$ for a fully explored trace set
$\mathcal{T}$. A direct consequence of the Definition~\ref{def:FE}
concerns the number of connected components and the stretch. (Recall
that the stretch is defined with respect to named nodes only, and
since $\alpha=1$, a 1-consistent inferrable topology cannot include
a shorter path between $u$ and $v$ than the one that must appear in
a trace of $\mathcal{T}$.)
\begin{lemma}\label{lemma:cc-fu-ex}
It holds that $\Comp(G_1)=\Comp(G_2)$ ($=\Comp(G_0)$) and the
stretch is 1.
\end{lemma}

The proof for the claims of the following lemmata are analogous to
our former proofs, as the main difference is the fact that there
might be more conflicts, i.e., edges in $G_*$.
\begin{lemma}\label{lemma:nofnodes_u}
For fully explored networks it holds that $|V_1|-|V_2|\leq s -
\gamma(G_*)\leq  s -1$ and $|V_1|/|V_2|\leq
(n+s)/(n+\gamma(G_*))\leq (2+s)/3$. Moreover, $|E_1|-|E_2|\in
2(s-\gamma(G_*))$ and $|E_1|/|E_2|\leq (\nu+2s)/(\nu+2)\leq s$,
where $\nu$ denotes the number of links between non-anonymous nodes.
There are traces with inferrable topology $G_1, G_2$ reaching these
bounds.
\end{lemma}

\begin{lemma}\label{lemma:degree_u}
For the maximal node degree, we have $\degree(G_1)-\degree(G_2)\leq
2(s-\gamma(G_*))$ and $\degree(G_1)/\degree(G_2)\leq
s-\gamma(G_*)+1$. There are instances $G_1,G_2$ that reach these
bounds.
\end{lemma}

From Lemma~\ref{lemma:cc-fu-ex} we know that fully explored
scenarios yield a perfect stretch of one. However, regarding the
diameter, the situation is different in the sense that distances
between anonymous nodes play a role.
\begin{lemma}\label{lemma:diam_u}
For connected topologies $G_1,G_2$ it holds that
$\diam(G_1)/\diam(G_2)\leq 2$, where $\diam$ denotes the graph
diameter and $\diam(G_1)>\diam(G_2)$. There are instances $G_1,G_2$
that reach this bound. Moreover, there are instances with
$\diam(G_1)-\diam(G_2)= s/2$.
\end{lemma}

The number of triangles with anonymous nodes can still not be
estimated accurately in the fully explored scenario.
\begin{lemma}\label{lemma:noftriangles_u}
There exist graphs where $C_3(G_1)-C_3(G_2)= s(s-1)/2$, and the
relative error $C_3(G_1)/C_3(G_2)$ can be arbitrarily large.
\end{lemma}

\section{Conclusion}\label{sec:conclusion}

We understand our work as a first step to shed light onto the
similarity of inferrable topologies based on most basic axioms and
without any assumptions on power-law properties, i.e., in the worst
case. Using our formal framework we show that the topologies for a
given trace set may differ significantly. Thus, it is impossible to
accurately characterize topological properties of complex networks.
To complement the general analysis, we propose the notion of fully
explored networks or trace sets, as a ``best possible scenario''. As
expected, we find that fully exploring traces allow us to determine
several properties of the network more accurately; however, it also
turns out that even in this scenario, other topological properties
are inherently hard to compute. Our results are summarized in
Figure~\ref{fig:sumres}.

\begin{footnotesize}
\begin{figure*}[t]\centering
\begin{footnotesize}
  \begin{tabular}{| l | c | c | c | c |}
    \hline    Property/Scenario &  \multicolumn{2}{c}{Arbitrary } & \multicolumn{2}{|c|}{Fully
      Explored ($\alpha=1$)}\\ \hline
    & $G_1 - G_2$ & $G_1/G_2$ &  $G_1 - G_2$ & $G_1/G_2$ \\
    \hline
    \# of nodes& $\leq s-\gamma(G_*) $& $\leq (n+s)/(n+\gamma(G_*)) $&$ \leq s-\gamma(G_*) $&$ \leq (n+s)/(n+\gamma(G_*))$   \\
    \hline
    \# of links& $\leq 2(s-\gamma(G_*)) $& $\leq (\nu+2s)/(\nu+2) $&$\leq 2(s-\gamma(G_*)) $& $\leq (\nu+2s)/(\nu+2) $   \\ \hline

    \# of connected components& $\leq n/2 $& $\leq n/2 $&$ =0 $& $=1$   \\
    \hline
    Stretch & - & $\leq (N-1)/2 $& - & $ =1 $   \\ \hline
    Diameter & $\leq (s-1)/s\cdot (N-1)$ & $\leq s$ & $s/2$ ($\P$) & $2$  \\ \hline
    Max. Deg. & $\leq 2(s-\gamma(G_*))$ &$\leq s-\gamma(G_*)+1$ & $\leq 2(s-\gamma(G_*))$ & $\leq s-\gamma(G_*)+1$  \\ \hline
    Triangles &$\leq 2s(s-1)$ &$\infty$ & $\leq 2s(s-1)/2$ & $\infty$  \\ \hline
  \end{tabular}
\end{footnotesize}
  \caption{\bf Summary of our bounds on the properties of inferrable
  topologies. $s$ denotes the number of stars in the traces, $n$ is the number of named nodes, $N=n+s$, and $\nu$ denotes the number of links between named nodes.
  Note that trace sets meeting these bounds exist for all properties for which we have tight or upper bounds.
  For the two entries marked with ($\P$), only ``lower bounds'' are
  derived, i.e., examples that yield at least the corresponding accuracy; as the upper bounds from the arbitrary scenario do not
  match, how to close the gap remains an open question.
} \label{fig:sumres}
\end{figure*}
\end{footnotesize}

Our work opens several directions for future research. On a
theoretical side, one may study whether the minimal inferrable
topologies considered in, e.g.,~\cite{icdcn10,icdcn11}, are more
similar in nature. More importantly, while this paper presented
results for the general worst-case, it would be interesting to
devise algorithms that compute, for a \emph{given trace set},
worst-case bounds for the properties under consideration. For
example, such approximate bounds would be helpful to decide whether
additional measurements are needed. Moreover, maybe such algorithms
may even give advice on the locations at which such measurements
would be most useful.

\section*{Acknowledgments}

We would like to thank H.~B.~Acharya and Steve Uhlig.


\bibliographystyle{plain}
\bibliography{gstar}


\begin{appendix}

\section{Deferred Proofs}

\subsection{Proof of Theorem~\ref{thm:gc-inferable}}

Fix $\mathcal{T}$. We have to prove that $G_C$ fulfills \RULEZERO,
\RULEONE\ (which implies \RULETHREE) and \RULETWO.

\RULEZERO: The axiom holds trivially: only edges from the traces are
used in $G_C$.

\RULEONE: Let $T\in \mathcal{T}$ and $\sigma_1,\sigma_2 \in T$. Let
$k=d_T(\sigma_1,\sigma_2)$. We show that $G_C$ fulfills \RULEONE,
namely, there exists a path of length $k$ in $G_C$. Induction on
$k$: ($k=1$:) By the definition of $G_C$, $\{\sigma_1,\sigma_2\} \in
E_C$ thus there exists a path of
  length one between $\sigma_1$ and $\sigma_2$.
($k> 1$:) Suppose \RULEONE\ holds up to $k-1$. Let
$\sigma'_1,\ldots,\sigma'_{k-1}$ be
  the intermediary nodes between $\sigma_1$ and $\sigma_2$ in $T$:
  $T=(\ldots,\sigma_1,\sigma'_1,\ldots,\sigma'_{k-1},\sigma_2,\ldots)$. By the induction hypothesis, in $G_C$ there is a path of
  length $k-1$ between $\sigma_1$ and $\sigma'_{k-1}$. Let $\pi$ be this path. By definition of
  $G_C$, $\{\sigma'_{k-1},\sigma_2\} \in E_C$. Thus appending $(\sigma'_{k-1},\sigma_2)$ to $\pi$ yields the desired path of length $k$
  linking $\sigma_1$ and $\sigma_2$: \RULEONE\ thus holds up to $k$.

\RULETWO: We have to show that $d_T(\sigma_1,\sigma_2) = k
\Rightarrow d_{C}(\sigma_1,\sigma_2)\geq\lceil \alpha\cdot k
\rceil$.
By contradiction, suppose that $G_C$ does not fulfill \RULETWO\ with
respect to $\alpha$. So there exists $k' < \lceil \alpha\cdot k
\rceil$ and $\sigma_1,\sigma_2 \in V_C$ such that
$d_C(\sigma_1,\sigma_2)=k'$. Let $\pi$ be a shortest path between
$\sigma_1$ and $\sigma_2$ in $G_C$. Let $(T_1, \ldots, T_{\ell})$ be
the corresponding (maybe repeating) traces covering this path $\pi$
in $G_C$. Let $T_i \in (T_1, \ldots, T_{\ell})$, and let $s_i$ and
$e_i$ be the corresponding start and end nodes of $\pi$ in $T_i$. We
will show that this path $\pi$ implies the existence of a path in
$G_0$ which violates $\alpha$-consistency. Since $G_0$ is
inferrable, $G_0$ fulfills \RULETWO, thus we have:
$d_C(\sigma_1,\sigma_2)=\sum_{i=1}^{\ell} d_{T_i}(s_i,e_i) = k' <
\lceil \alpha\cdot k \rceil \leq d_{G_0}(\sigma_1,\sigma_2)$ since
$G_0$ is $\alpha$-consistent. However, $G_0$ also fulfills \RULEONE,
thus $d_{T_i}(s_i,e_i) \geq d_{G_0}(s_i,e_i)$. Thus
$\sum_{i=1}^{\ell} d_{G_0}(s_i,e_i) \leq \sum_{i=1}^{\ell}
d_{T_i}(s_i,e_i)< d_{G_0}(\sigma_1,\sigma_2)$: we have constructed a
path from $\sigma_1$ to $\sigma_2$ in $G_0$ whose length is shorter
than the distance between $\sigma_1$ and $\sigma_2$ in $G_0$,
leading to the desired contradiction.

\subsection{Proof of Lemma~\ref{lemma:allstars}}

First we construct a topology $G_0=(V_0,E_0)$ and then describe a
trace set on this graph that generates the star graph $G=(V,E)$. The
node set $V_0$ consists of $|V|$ anonymous nodes and $|V|\cdot
(1+\tau)$ named nodes, where $\tau=\lceil 3/(2\alpha) - 1/2\rceil$.
The first building block of $G_0$ is a copy of $G$. To each node
$v_i$ in the copy of $G$ we add a chain consisting of $2+\tau$
nodes, first appending $\tau$ non-anonymous nodes $w_{(i,k)}$ where
$1 \leq k \leq \tau$, followed by an anonymous node $u_i$ and
finally a named node $w_{(i,\tau+1)}$. More formally we can describe
the link set as $E_0 = E \cup \bigcup_{i=1}^{|V|} \left(
\{v_i,w_{(i,1)}\}, \{w_{(i,1)}, w_{(i,2)}\},\ldots,
\{w_{(i,\tau)},u_i\}, \{u_i,w_{(i,\tau+1)}\}\right)$. The trace set
$\mathcal{T}$ consists of the following $|V|+|E|$ shortest path
traces: the traces $T_{\ell}$ for $\ell\in\{1,\ldots, |V|\}$, are
given by $T_{\ell}(w_{(\ell,\tau)},w_{(\ell,\tau+1)})$ (for each
node in $V$), and the traces $T_{\ell}$ for
${\ell}\in\{|V|+1,\ldots, |V|+|E|\}$, are given by
$T_{\ell}(w_{(i,\tau)},w_{(j,\tau)})$ for each link $\{v_i,v_j\}$ in
$E$. Note that $G_0=G_C$ as each star appears as a separate
anonymous node. The star graph $G_*$ corresponding to this trace set
contains the $|V|$ nodes $*_i$ (corresponding to $u_i$). In order to
prove the claim of the lemma we have to show that two nodes
$*_i,*_j$ are conflicting according to
Lemma~\ref{lem:lemmaInferencemulti} if and only if there is a link
$\{v_i,v_j\}$ in $E$.
Case~$(i)$ does not apply because the minimum distance between any
two nodes in the canonic graph is at least one, and $\lceil
\alpha\cdot d_{T_i}(*_i,w_{(i,\tau)})\rceil=1$ and $\lceil
\alpha\cdot d_{T_i}(*_i,w_{(i,\tau+1)})\rceil=1$. It remains to
examine Case~$(ii)$: ``$\Rightarrow$'' if $\map(*_i)=\map(*_j)$
there would be a path of length two between $w_{(i,\tau)}$ and
$w_{(j,\tau)}$ in the topology generated by $\map$; the trace set
however contains a trace $T_\ell(w_{(i,\tau)},w_{(j,\tau)})$ of
length $2\tau+1$. So $\lceil \alpha\cdot d_{T_\ell}
(w_{(i,\tau)},w_{(j,\tau)})\rceil =\lceil\alpha\cdot(2\tau+1)\rceil
=\lceil \alpha\cdot(2\lceil 3/(2\alpha) - 1/2\rceil +1\rceil) \geq
3$, which violates the $\alpha$-consistency
(Lemma~\ref{lem:lemmaInferencemulti}~(ii)) and hence $\{*_i,*_j\}\in
E_*$ and $\{v_i,v_j\}\in E$. ``$\Leftarrow$'': if
$\{v_i,v_j\}\not\in E$, there is no trace
$T(w_{(i,\tau)},w_{(j,\tau)})$, thus we have to prove that no trace
$T_\ell(w_{(i',\tau)},w_{(j',\tau)})$ with $i'\neq i$ and $j'\neq j$
and $j'\neq i$ leads to a conflict between $*_i$ and $*_j$.  We show
that an even more general statement is true, namely that for any
pair of distinct non-anonymous nodes $x_1,x_2$, where
$x_1,x_2\in\{v_{i'},v_{j'},w_{(i',k)}, w_{(j',k)}|1 \leq k \leq
\tau+1, i'\neq i, j'\neq i, j'\neq j\}$, it holds that
$\lceil\alpha\cdot d_{C}(x_1,x_2)\rceil\leq
d_C(x_1,*_i)+d_C(x_2,*_j)$. Since $G_C=G_0$ and the traces contain
shortest paths only, the trace distance between two nodes in the
same trace is the same as the distance in $G_C$. The following
tables contain the relevant lower bounds on distances in $G_C$ and
$\mu(x_1,x_2)=d_C(x_1,*_i)+d_C(x_2,*_j)$.

\begin{footnotesize}
\begin{table}[h]
    \centering
        \begin{tabular}{ c | c  c  c  c }
            $d_{C}(\cdot,\cdot)\geq$ &  $v_{i'}$ & $v_{j'}$ & $w_{(i',k_1)}$ & $w_{(j',k_1)} $\\
            \hline
            $v_{i'}$ & 0 & 1 & $k_1$ & $k_1 + 1$\\
           $v_{j'}$ & 1 & 0 & $k_1$ + 1 & $k_1$\\
           $w_{(i',k_2)}$ & $k_2$ & $k_2 + 1$ & $|k_2 - k_1|$ & $k_1 + 1 + k_2$\\
            $w_{(j',k_2)}$ &$k_2 + 1$ & $k_2$ & $k_1 + 1 + k_2$ & $|k_2 - k_1|$\\
            $*_{i}$ & $ \tau+2$ & $ \tau+1$ & $ 2 + \tau + k_1$ & $\tau-k_1+1$\\
            $*_{j}$ & $ \tau+2$ & $ \tau+2$ & $ 2 + \tau + k_1$ & $ 2 + \tau + k_1$\\
         \multicolumn{5}{c}{~~}\\
                $\mu(\cdot,\cdot)\geq$ &    $v_{i'}$ & $v_{j'}$ & $w_{(i',k_1)}$ & $w_{(j',k_1)} $\\
            \hline
            $v_{i'}$ & $2\tau+4$ & $ 2\tau+3$ & $ 4+2\tau+k_1$ & $4+2\tau+k_1$\\
            $v_{j'}$ & $2\tau+3$ & $2\tau+4$ & $ 2\tau+3+k_1$ & $3+2\tau+k_1$\\
            $w_{(i',k_2)}$ & $ 4 + 2\tau + k_2$ & $ 4 + 2\tau + k_2$ & $4+2\tau+k_1+k_2$ & $ 4 + 2\tau + k_1+k_2$\\
            $w_{(j',k_2)}$ & $2\tau-k_2+3$ & $2\tau-k_2+3$ & $2\tau+3+k_1-k_2$ & $2\tau+k_1-k_2+3$\\
        \end{tabular}
        \caption{\bf Proof of Lemma~\ref{lemma:allstars}: lower bounds for the distances in $G_C$, and lower bounds for $\mu(x_1,x_2)=d_C(x_1,*_i)+d_C(x_2,*_j)$.}
\end{table}
\end{footnotesize}

If $x_1\neq w_{(j',k_2)}$ then it holds for all $x_1,x_2$ that
$d_{T_\ell} (x_1,x_2) \leq2 \tau+1$ whereas $\mu(x_1,x_2) =
d_C(x_1,*_i) + d_C(x_2,*_j) \geq 2 \tau+2$. In all other cases it
holds at least that $d_{C}(x_1,x_2) < \mu(x_1,x_2)$. Thus
$\lceil\alpha\cdot d_{C}(x_1,x_2)\rceil\leq
d_C(x_1,*_i)+d_C(x_2,*_j)$. Consequently, we have conflicts if and
only if $\{v_i,v_j\}\in E$, which concludes the proof.

\subsection{Proof of Lemma~\ref{thm:bijection}}

\begin{wrapfigure}{l}{0.42\textwidth}
        \vspace*{-.0cm}
    \includegraphics[width=0.38\textwidth]{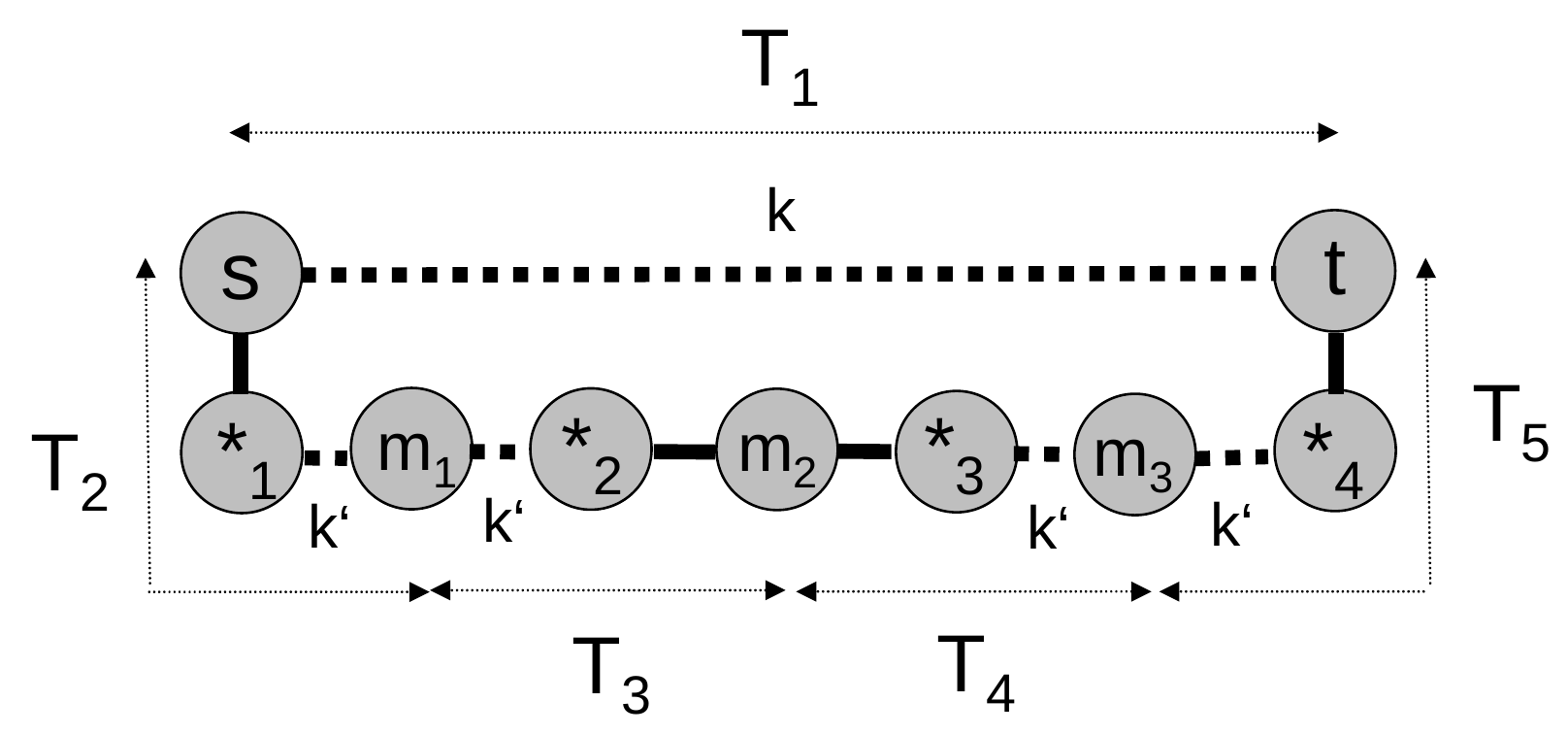}\caption{\bf Visualization for proof of Lemma~\ref{thm:bijection}.
    Solid lines denote links, dashed lines denote paths (of annotated length).}\label{fig:gilles-pic}
    \vspace*{-.5cm}
\end{wrapfigure}

We have to show that the paths in the traces correspond to paths in
$G_{\gamma}$. Let $T \in \mathcal{T}$, and $\sigma_1,\sigma_2 \in
T$. Let $\pi$ be the sequence of nodes in $T$ connecting $\sigma_1$
and $\sigma_2$. This is also a path in $G_{\gamma}$: since
$\alpha>0$, for any two symbols $\sigma_1,\sigma_2 \in T$, it holds
that $\map(\sigma_1)\neq \map(\sigma_2)$ as $\alpha>0$.


%

We now construct an example showing that the $\alpha'$ for which
$G_{\gamma}$ fulfills \RULETWO\ can be arbitrarily small. Consider
the graph represented in Figure~\ref{fig:gilles-pic}. Let
$T_1=(s,\ldots,t), T_2=(s,*_1,\ldots,m_1), T_3=(m_1,\ldots,*_2,m_2),
T_4=(m_2,*_3,\ldots,m_3), T_5=(m_3,\ldots, *_4,t)$. We assume
$\alpha=1$. By changing parameters $k=d_C(s,t)$ and
$k'=d_C(m_1,*_1)=d_C(m_1,*_2)=d_C(m_3,*_3)=d_C(m_3,*_4)$, we can
modulate the links of the corresponding star graph $G_*$.  Using
$d_{T_1}(s,t)=k$, observe that $k>2 \Leftrightarrow \{*_1,*_4\}\in
E_*$. Similarly, $k>2(k'+1) \Leftrightarrow \{*_1,*_3\}\in E_*
\wedge \{*_2,*_4\}\in E_*$ and $k>2(k'+2) \Leftrightarrow
\{*_1,*_2\}\in E_* \wedge \{*_3,*_4\}\in E_*$. Taking $k=2k'+4$, we
thus have $E_*=\{\{*_1,*_3\},\{*_2,*_4\},\{*_1,*_4\}\}$.

Thus, we here construct a situation where $*_1$ and $*_2$ as well as
$*_3$ and $*_4$ can be merged without breaking the consistency
requirement, but where merging both simultaneously leads to a
topology $G'$ that is only $4/k$-consistent, since $d_{G'}(s,t)=4$.
This ratio can be made arbitrarily small provided we choose
$k'=(k-4)/2$.

\subsection{Proof of Lemma~\ref{lemma:nofnodesedges}}

In the worst-case, each star in the trace represents a different
node in $G_1$, so the maximal number of nodes in any topology in
$\mathcal{G}_{\mathcal{T}}$ is the total number of non-anonymous
nodes plus the total number of stars in $\mathcal{T}$. This number
of nodes is reached in the topology $G_C$. According to
Definition~\ref{def:stargraph}, only non-adjacent stars in $G_*$ can
represent the same node in an inferrable topology. Thus, the stars
in trace $\mathcal{T}$ must originate from at least $\gamma(G_*)$
different nodes. As a consequence $|V_1|-|V_2|\leq s - \gamma(G_*)$,
which can reach $s-1$ for a trace set
$\mathcal{T}=\{T_i=(v,*_i,w)|1\leq i\leq s\}$. Analogously,
$|V_1|/|V_2|\leq (n+s)/(n+\gamma(G_*))\leq (2+s)/3$.

Observe that each occurrence of a node in a trace describes at most
two edges. If all anonymous nodes are merged into $\gamma(G_*)$
nodes in $G_1$ and are separate nodes in $G_2$ the difference in the
number of edges is at most $2(s-\gamma(G_*))$. Analogously,
$|E_1|/|E_2|\leq (\nu+2s)/(\nu+2)\leq s$. The trace set
$\mathcal{T}=\{T_i=(v,*_i,w)|1\leq i\leq s\}$ reaches this bound.

\subsection{Proof of Lemma~\ref{lemma:distortion}}

An ``lower bound'' example follows from
Figure~\ref{fig:consistency}. Essentially, this is also the worst
case: note that the difference in the shortest distance between a
pair of nodes $u$ and $v$ in $G_1$ and $G_2$ is only greater than 0
if the shortest path between them involves at least one anonymous
node. Hence the shortest distance between such a pair is two. The
longest shortest distance between the same pair of nodes in another
inferred topology visits all nodes in the network, i.e., its length
is bounded by $N-1$.

\subsection{Proof of Lemma~\ref{lemma:degree}}

Each occurrence of a node in a trace describes at most two links
incident to this node. For the degree difference we only have to
consider the links incident to at least one anonymous node, as the
number of links between non-anonymous nodes is the same in $G_1$ and
$G_2$. If all anonymous nodes can be merged into $\gamma(G_*)$ nodes
in $G_1$ and all anonymous nodes are separate in $G_2$ the
difference in the maximum degree is thus at most $2(s-\gamma(G_*))$,
as there can be at most $s-\gamma(G_*)+1$ nodes merged into one node
and the minimal maximum degree of a node in $G_2$ is two. This bound
is tight, as the trace set $T_i=\{v_i,*,w_i\}$ for $1\leq i\leq s$
containing $s$ stars can be represented by a graph with one
anonymous node of degree $2s$ or by a graph with $s$ anonymous nodes
of degree two each. For the ratio of the maximal degree we can
ignore links between non-anonymous nodes as well, as these only
decrease the ratio. The highest number of links incident at node $v$
with one endpoint in the set of anonymous nodes is $s-\gamma(G_*)+1$
for non-anonymous nodes and $2(s-\gamma(G_*)+1)$ for anonymous
nodes, whereas the lowest number is two.

\subsection{Proof of Lemma~\ref{lemma:degree_u}}
The proof for the upper bound is analogous to the case without full
exploration. To prove that this bound can be reached, we need to
add traces to the trace set to ensure that all pairs of named nodes appear in
the trace but does not change the degrees of anonymous nodes. To this end we
add a named node $u$ for each pair $\{v,w\}$ that is not in the trace set yet
to $G_0$ and a trace $T=\{v,u,w\}$. This does not increase the maximum degree
and guarantees full exploration.

\subsection{Proof of Lemma~\ref{lemma:diam_u}}

We first prove the upper bound for the relative case. Note that the
maximal distance between two anonymous nodes $\map(*_1)$ and
$\map(*_2)$ in an inferred topology component cannot be larger than
twice the distance of two named nodes $u$ and $v$: from
Definition~\ref{def:FE} we know that there must be a trace in
$\mathcal{T}$ connecting $u$ and $v$, and the maximal distance
$\delta$ of a pair of named nodes is given by the path of the trace
that includes $u$ and $v$. Therefore, and since any trace starts and
ends with a named node, any star can be at a distance at a distance
$\delta/2$ from a named node. Therefore, the maximal distance
between $\map(*_1)$ and $\map(*_2)$ is $\delta/2+\delta/2$ to get to
the corresponding closest named nodes, plus $\delta$ for the
connection between the named nodes. As according to
Lemma~\ref{lemma:cc-fu-ex}, the distance between named nodes is the
same in all inferred topologies, the diameter of inferred topologies
can vary at most by a factor of two.

We now construct an example that reaches this bound. Consider a
topology consisting of a center node $c$ and four rays of length
$k$. Let $u_1, u_2, u_3, u_4$ be the ``end nodes'' of each ray. We
assume that all these nodes are named. Now add two chains of
anonymous nodes of length $2k+1$ between nodes $u_1$ and $u_2$, and
between nodes $u_3$ and $u_4$  to the topology. The trace set
consists of the minimal trace set to obtain a fully explored
topology: six traces of length $2k+1$ between each pair of end nodes
$u_1, u_2, u_3, u_4$. Now we add two traces of length $2k+1$ between
nodes $u_1$ and $u_2$, and between nodes $u_3$ and $u_4$. These
traces explore the anonymous chains and have the following shape:
$T_7=(u_1,
*_1,\ldots,*_k,\sigma,*_{k+1},\ldots,*_{2k}, u_2)$ and $T_8=(u_3,
*_{2k+1},\ldots,*_{3k},\sigma',*_{3k+1},\ldots,*_{4k}, u_4)$, where
$\sigma$ and $\sigma'$ are stars. Let $G_1=G_C$ and $G_2$ be the
inferrable graph where $\sigma$ and $\sigma'$ are merged. The
resulting diameters are $\diam(G_1)=4k+2$ and $\diam(G_2)=2k+1$.
Since $s=4k+2$, the difference can thus be as large as $s/2$. Note
that this construction also yields the bound of the relative
difference: $\diam(G_1)/\diam(G_2)=(4k+2)/(2k+1)=2$.

\subsection{Proof of Lemma~\ref{lemma:noftriangles_u}}

Given the number of stars $s$, we construct a trace set
$\mathcal{T}$ with two inferrable graphs such that in one graph the
number of triangles with anonymous nodes is $s(s-1)/2$ and in  the
other graph there are no such triangles. As a first step we add $s$
traces $T_i=(v_i,*_i,w)$ to the trace set $\mathcal{T}$, where
$1\leq i\leq s$. To make this trace set fully explored we add traces
for each pair $v_i,v_j$ to $\mathcal{T}$ as a second step, i.e.,
traces $T_{i,j} = (v_i,v_j)$ for $1\leq i\leq s$ and $1\leq j \leq
s$. The resulting trace set contains $s$ stars and none of the stars
are in conflict with each other. Thus the graph $G_1$ merging all
stars into one anonymous node is inferrable from this trace and the
number of triangles where the anonymous node is part of is
$s(s-1)/2$. Let $G_2$ be the canonic graph of this trace set. This
graph does not contain any triangles with anonymous nodes and hence
the difference $C(G_1)-C(G_2)$ is $s(s-1)/2$.

To see that the ratio can be unbounded look at the trace set
$\{(v,*_1,w),(u,*_2,w),(u,v)\}$. This set is fully explored since
all pairs of named nodes appear in a trace. The graph where the two
stars are merged has one triangle and the canonic graph has no
triangle.
\end{appendix}

\end{document}